\documentclass[aps,pra,twocolumn,showpacs,preprintnumbers,superscriptaddress]{revtex4}
\usepackage{amsmath}
\usepackage{amsfonts}
\usepackage{amsthm}
\usepackage{amssymb}
\usepackage{url}
\usepackage{ifthen}
\usepackage{color}
\usepackage{graphicx}
\usepackage{float}
\newboolean{qcircuit}

\setboolean{qcircuit}{true}     

\newcommand{\comment}[1]{ }
\newcommand{\ba}{\begin{eqnarray}}
\newcommand{\ea}{\end{eqnarray}}

\newtheorem{definition}{Definition}

\newtheorem{theorem}{Theorem}

\newtheorem{lemma}{Lemma}

\begin{document}

\title{Min-entropy sources for Bell tests}

\author{Le Phuc Thinh}
\affiliation{Centre for Quantum Technologies, National University of Singapore, 2 Science Drive 3, Singapore 117543}
\author{Lana Sheridan}
\affiliation{Centre for Quantum Technologies, National University of Singapore, 2 Science Drive 3, Singapore 117543}
\author{Valerio Scarani}
\affiliation{Centre for Quantum Technologies, National University of Singapore, 2 Science Drive 3, Singapore 117543}
\affiliation{Department of Physics, National University of Singapore, 3 Science Drive 2, Singapore 117542}

\begin{abstract}
Device independent protocols rely on the violation of Bell inequalities to certify properties of the resources available.  The violation of the inequalities are meaningless without a few well-known assumptions.  One of these is measurement independence, the property that the source of the states measured in an inequality is uncorrelated from the measurements selected.  Since this assumption cannot be confirmed, we consider the consequences of relaxing it and find that the definition chosen is critically important to the observed behavior.  Considering a definition that is a bound on the min-entropy of the measurement settings, we find lower bounds on the min-entropy of the source used to choose the inputs required to deduce any quantum or non-local behavior from a Bell inequality violation.  These bounds are significantly more restrictive than the ones obtained by endowing the measurement-input source with the further structure of a Santha-Vazirani source.  We also outline a procedure for finding tight bounds and study the set of probabilities that can result from relaxing measurement dependence. 
\end{abstract}


\maketitle



\section{Introduction}
\label{sec:intro}

The violation of Bell inequalities can be used to certify important quantum information properties in a black-box scenario under minimal assumptions. This idea of ``device-independent" certification started in the context of quantum key distribution, where the violation of Bell inequalities bounds the information leaked to the eavesdropper~\cite{MY98,ABGM+07,VV12b}; and it has been extended to various other tasks, notably state certification~\cite{MY98,BLM+09,MYS12}, measurement certification~\cite{RHC+11}, and private randomness expansion~\cite{PAM10,CK11,VV12}. Ultimately, this stems from the fact that the violation of Bell inequalities certifies the presence of a quantifiable amount of intrinsic randomness: indeed, \textit{a contrario}, if the outcomes were predictable, one could have predicted them in advance and the measurement could consist of reading from a pre-existing list.  This is exactly what the violation of Bell inequality certifies as impossible.

Two assumptions are left in device-independent certification. The first is \textit{no-signaling}: the choice of the measurement setting of one party should not be known to the measurement boxes of the other parties before they produce their outcome. This can be guaranteed ultimately by ensuring space-like separation, although one may also trust a weaker demonstration of separation, as for instance in \cite{PAM10}. The second assumption is \textit{measurement independence}: the information $\lambda$ contained in the boxes in each run should be uncorrelated from the choice of the settings in that run. So far, no way of checking measurement independence is known in a black-box scenario: the best one can do is to buy the source of $\lambda$ and the devices that choose the settings from different providers, who are believed not to be conspiring together. Alternatively, one can partly give up the black-box scenario, characterize the devices and be confident that the relevant degrees of freedom are uncorrelated.

It is clear that no-signaling and measurement independence cannot be arbitrarily relaxed: if any amount of signaling is allowed, or if arbitary correlation is admitted between source and settings, the violation of a Bell inequality can be obtained with purely classical resources $\lambda$, so there is no hope to conclude that $\lambda$ contains intrinsic randomness. However, with \textit{the aim of reducing the assumptions of device-independent certification to their bare minimum}, one can \textit{partially} relax no-signaling and measurement independence, and ask how much information must be signaled and how much measurement dependence must be allowed for a Bell test to become irrelevant \cite{H11}. In this paper, we focus on the latter question, the study of \textit{partial measurement dependence} (sometimes called reduced measurement independence or reduced ``free will"), which has been the object of a few recent studies~\cite{CR12,KHSPMKSE12,GMTDAA12,MP13}.  In particular, we consider the random source that is required to choose the input settings for a Bell inequality and place bounds on the min-entropy necessary to show any difference between local and no-signaling output distributions. Note that if the violation of a Bell inequality is used in a device independent protocol to certify the amplification or expansion of input randomness, this source would serve as the seed randomness in the protocol.

\section{Measurement dependence and its basic consequences}

\subsection{Measurement independence}
\label{sec:measind}

For the sake of this introduction, we consider a bipartite Bell scenario. Operationally, a Bell experiment consists of $N$ apparently identical runs \footnote{We focus on the operational description of current experiments and do not consider the more general, but as yet abstract, case of \textit{parallel repetition}, in which all the inputs are given at the same time.}, in each of which box A receives input $x$ and outputs a value $a$, box B receives input $y$ and outputs a value $b$.  A \emph{measurement-setting source} (henceforth \emph{source}) for the Bell test supplies the experimentalist with inputs $x$ and $y$; its behavior is modelled by a probability distribution $p(xy|\lambda)$. One can then estimate the statistics $p(ab|xy)$. We denote by $\lambda$ the information present in the boxes in a given run.

Measurement independence, the assumption that we want to relax, is captured by the condition
\begin{equation}
\label{eq:micondition}
p(\lambda|xy) = p(\lambda) \ \forall x,y\,.
\end{equation} Under this assumption, the observed statistics are modeled by
\begin{equation}
\label{eq:mimodel}
p_\text{MI}(ab|xy) = \int p(ab|xy \lambda)p(\lambda) \, \text{d}\lambda \ .
\end{equation} The specific goal of a Bell test is to assess whether there is intrinsic randomness in the boxes, that is, in the usual terminology, to guarantee that $\lambda$ is not a \textit{local variable}. Mathematically, local variables are defined by $p(ab|xy \lambda)=p(a|x\lambda)p(b|y\lambda)$. It is useful to stress that, as written, \eqref{eq:mimodel} contains an additional assumption, namely that $\lambda$ itself is chosen independently in each run according to the distribution $p(\lambda)$. Under measurement independence, it can be proved that this is ultimately not a restriction for Bell tests, although one has to be careful in interpreting statistics from finite samples \cite{BCHK+02,Gill03,ZGK11}.

Measurement independence cannot be denied in a systematic way without undermining the scientific method itself (if a clinical trial is to make sense, whether each patient receives the drug or the placebo cannot depend on the any details of the patients' conditions). However, it is certainly possible to question measurement independence in a given setup: the devices that determine the inputs $x,y$ may be correlated to the process that determines $\lambda$. The origin of such correlation may be trivial, like the fluctuations in power of the city network to which all the devices are connected; it may be due to lack of attention of the experimentalists, who introduced unwanted connections; or it may be strongly conspiratorial, in an adversarial scenario in which the devices come from an untrusted provider. In all cases, \eqref{eq:micondition} does not hold, nor does the proof that one can restrict the study to independently-chosen $\lambda$.

\begin{figure}[h]
\includegraphics[scale=0.25]{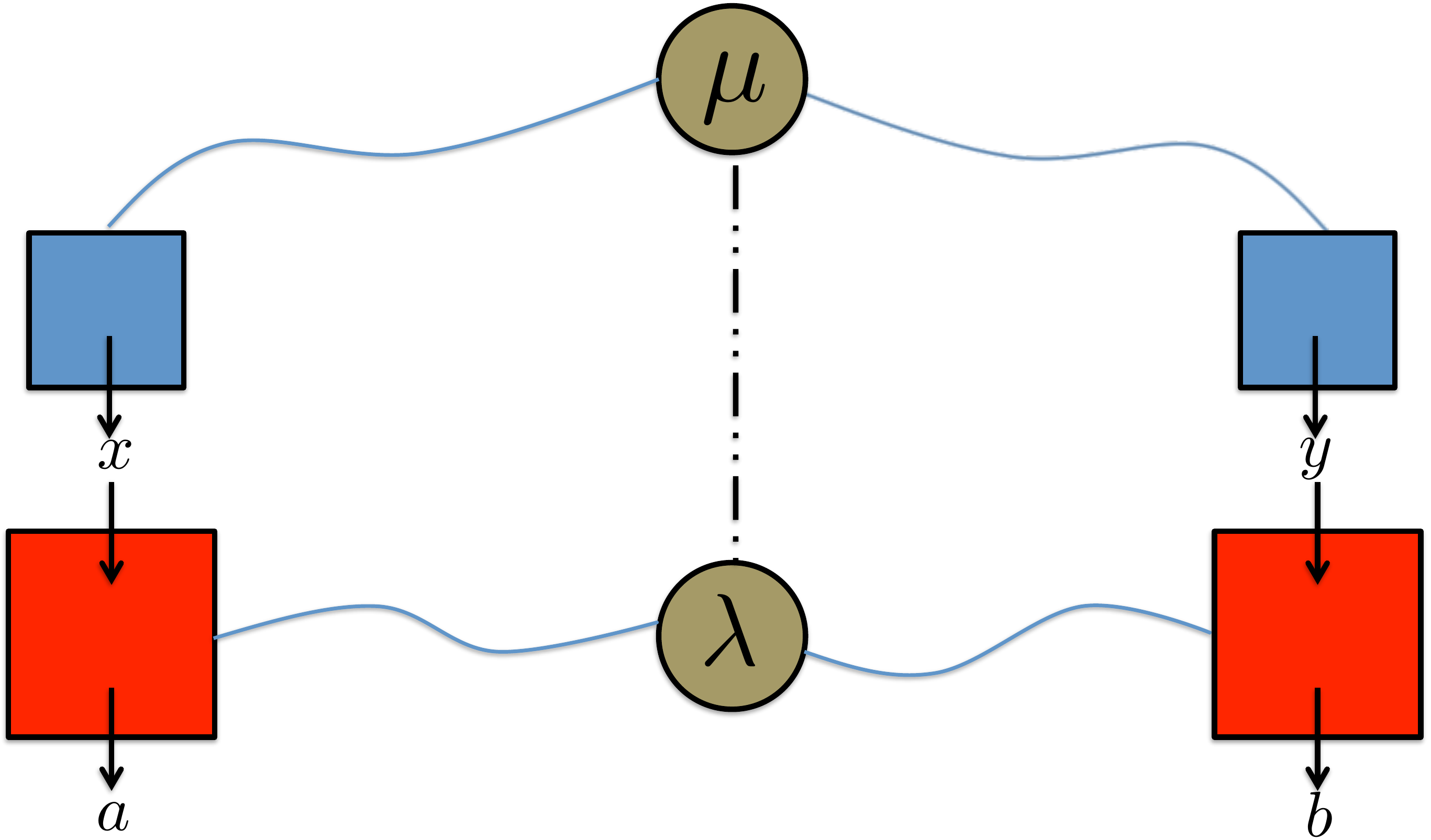}
\caption{(color online) There are many different processes by which the information $\lambda$ in devices $A$ and $B$ might become correlated with the inputs to the devices $x$ and $y$, as discussed in the text.  In this illustration the processes are represented by some external pre-existing variable $\mu$ that serves to introduce the correlation.  The blue boxes represent the physical random number generators used to pick the inputs to the Bell test.}
\label{fig:bipartiteillustration}
\end{figure}

By relaxing condition \eqref{eq:micondition}, one allows correlations between the boxes' content $\lambda$ and the choice of the settings $x,y$. Bayes theorem implies that \begin{equation}
p(\lambda|xy) \neq p(\lambda)\Longleftrightarrow p(xy|\lambda) \neq p(xy)\,.
\end{equation} The first relation could be read as ``the output of the source is restricted for a given choice of settings", the second as ``the choice of settings is restricted for a given output of the source". Neither needs to refer to a real causal relation: all is compatible with both $\lambda$ and $x,y$ being influenced by a common cause (Fig.~\ref{fig:bipartiteillustration}). That being clarified, our discourse will be mostly phrased in the second way (the first way will be used in Section \ref{sec:lambda}). We shall then look at measurement dependence as \textit{reducing the probability of certain pairs of settings}. In the case where the dependence is sufficient to \textit{exclude} enough pairs of settings, unwanted features of local variable models may be hidden. This is the same intuition behind the power of the detection loophole; in fact, measurement dependence is even stronger, because it may allow to exclude a single \textit{pair} of settings, whereas the detection loophole is local and excludes all pairs of settings such that one given setting of (say) Bob is associated to unwanted features. This opens a wealth of possibilities that we review rapidly next.

\subsection{Effects of measurement dependence}
\label{ssviol}

The obvious effect of measurement independence is the possibility of \textit{faking a violation of Bell inequalities}. A Bell inequality is built on a linear combination of $p(ab|xy)$, whose maximal value (called \textit{algebraic limit}) cannot be reached by local variables. If, in each run, one can exclude some suitable pairs of settings in correlation with the content of the boxes $\lambda$, then it becomes possible to reach the algebraic limit while having only local variables in the boxes.

Let us illustrate this point with the most famous Bell inequality, that of Clauser, Horne, Shimony and Holt (CHSH). The inequality reads
\begin{equation}
\left| \langle a_0 b_0 \rangle + \langle a_0 b_1 \rangle + \langle a_1 b_0 \rangle - \langle a_1 b_1 \rangle \right| \leq 2
\end{equation} with $a_x,b_y\in\{-1,+1\}$. In order to achieve the algebraic limit of $4$, one should have $a_0=b_0$, $a_1=b_0$, $a_0=b_1$ and $a_1=-b_1$. Local deterministic points exist that satisfy three out of these four conditions. If one wants to achieve the algebraic limit with local variable and measurement dependence, a sufficient strategy is the following: in each run, $\lambda$ is chosen among the aforementioned local deterministic points, and the pair of settings corresponding to the unwanted condition is never chosen~\cite{H11,KHSPMKSE12}.

The fact that a sufficient amount of measurement dependence can lead to the algebraic limit has an intriguing consequence for some inequalities. Indeed, in generic inequalities, the algebraic limit may lie even above what can be reached with no-signaling correlations. For instance, the tilted CHSH inequality
\begin{equation}
\left| \langle a_0 b_0 \rangle + \langle a_0 b_1 \rangle + \langle a_1 b_0 \rangle - \langle a_1 b_1 \rangle + \alpha \langle a_0 \rangle \right| \leq 2+ \alpha \ ,
\end{equation}
has an algebraic limit of $4+\alpha$, but no-signaling correlations can reach only up to 4 if $\alpha\leq 2$~\cite{AMP12}. If measurement dependence is allowed, to the point that one pair of settings can be excluded, then one can achieve the algebraic limit with a convex mixture of
\begin{align} 
\label{sigstrategy} 
\begin{array}{ccc}
\lambda = (+1, -1, -1, +1) & \textrm{together with}&  (x,y) \neq 00 \\
\lambda = (+1, +1, +1, -1) & \textrm{together with}&  (x,y) \neq 01 \\
\lambda = (+1, -1, +1, +1) & \textrm{together with}&  (x,y) \neq 10 \\
\lambda = (+1, +1, +1, +1) & \textrm{together with}& (x,y) \neq 11
\end{array}
\end{align} where we denoted a local deterministic point as $\lambda=(a_0,a_1,b_0,b_1)$. If a Bell test is run with this underlying strategy, the observed correlations will lie outside the no-signaling polytope, i.e. are formally signaling. Obviously, this does not mean that measurement dependence makes it possible to use entanglement to actually send a message: in order for (say) Alice to send a message to Bob, she must be able to choose her setting at will, which is precisely what measurement dependence denies. At any rate, one must be careful when working with measurement dependence: the worst case are correlations that reach the algebraic limit, not the no-signaling one (to our knowledge, all the studies of measurement dependence so far dealt with inequalities for which the two limits happen to coincide \cite{H11,CR12,KHSPMKSE12,GMTDAA12,MP13}).

The take-away message of this paragraph is that one does not have to reach the extreme case of total measurement dependence (i.e.~$\lambda$ determining $x,y$ uniquely): \textit{already with some partial amount of measurement dependence, it becomes impossible to draw any conclusion from the violation of a Bell inequality}. This has important consequences when the source is characterized only by its conditional min-entropy. Indeed, one of our main result will consist in deriving general bounds for this amount (Section \ref{sec:robustness}). In order to do that, we need first to recall the definition of min-entropy and its relation to the Santha-Vazirani condition in light of measurement dependence.

\section{Min-entropy and measurement dependence}
\label{sec:defs}

As mentioned, the source of the Bell test behaves according to $p(xy|\lambda)$. Measurement independence implies that $p(xy|\lambda)$ has as much entropy or randomness as $p(xy)$. In contrast, partial measurement dependence means that there is some randomness in the source, but it is less than the entropy of the distribution $p(xy)$. The min-entropy and min-entropy deficit $H_{\min}(\mathbf{Z})-H_{\min}(\mathbf{Z}|\mathbf{\Lambda})$ are measures of randomness of a source, and they partly capture the amount of measurement dependence in special cases. But note that they are not intrinsic measures of measurement dependence (for instance, min-entropy deficit equals 0 does not imply measurement independence). If the min-entropy is not high enough, it leaves open the possibility of excluding certain settings, which allows faking of Bell violations as we discussed before. This behavior is forbidden in Santha-Vazirani sources as explained next.

\subsection{Min-entropy vs Santha-Vazirani condition}
\label{ss:critical}

We illustrate our point with an example. The \textit{chained inequality} is a bipartite Bell inequality with $m$ settings for each party and binary outcomes $a,b$ for both measurements on $A$ and $B$, which reads
\ba
I_m &=& p(a = b | x=1, y=m) + \sum_{\stackrel{x,y \text{ s.t.}}{x\in \{y,y+1\}}} p(a \neq b | x, y)  \nonumber \\
&\leq& 2m-1\,.
\ea
It has been used to put stringent bounds on quantum theory thanks to the property that, in the limit $m\rightarrow\infty$, its algebraic limit $I_m=2m$ can be reached with measurements on quantum states \cite{BKP06,CR08}.

Out of the $m^2$ possible pairs of settings, $2m$ are effectively used in the inequality. Furthermore, there exist local deterministic points that can satisfy $2m-1$ of these conditions. Therefore, in order to verify any conclusion based on the chained inequality, it is enough to have an amount of measurement dependence that allows the exclusion of only one pair of settings out of $m^2$. In the limit of large $m$, under whichever measure, such a source is very close to a fully random source: for instance, its min-entropy per run (defined below) is $\log(m^2-1)$, which differs from the fully random value $\log m^2$ by $O(m^{-2})$. This example shows that \textit{a source, which would presumably be considered as good as it gets in an abstract assessment, is already catastrophic for the Bell inequality under study}. Notice that this remark is not in contradiction with the results of~\cite{CR12}, which can be seen as proving that the chained inequality is pretty robust to measurement dependence: indeed, in that work, the additional Santha-Vazirani assumption was made on the source, which implies that all the pairs of settings are possible in each run. Our argument, based on excluding one setting in each run, does not apply.

It is now time to present the definitions we have just sketched in their suitable formal setting. We shall consistently use the word \textit{source} to stress that the source of randomness we are interested in is the randomness of the inputs given the knowledge of the physical process $\lambda$ or vice versa, not the randomness possibly present in $\lambda$ (which would be the intrinsic randomness of quantum origin in the ideal case).

\subsection{Formal definitions}

Here we review rapidly the definitions of well-known types of sources of randomness for the purpose of this paper, referring to~\cite{Vadhan} for a comprehensive study.

Consider a random variable $Z$ in an alphabet $\mathcal{Z}$ of size $d$; and let $\mathbf{Z} = Z_1...Z_N$ be an $N$-dit string. In our case, $Z$ will represent the settings chosen for the Bell test, i.e. $Z=(x,y)$ in a bipartite scenario. Randomness being synonymous with unpredictability, a source of randomness will be characterized by specifying what one wants to predict and how predictable it is, given some prior information $\mathbf{\Lambda}$ (supposed to be classical throughout this paper). One would then say that the source contains randomness if
\begin{equation}
P_{\text{guess}}(\mathbf{Z}|\mathbf{\Lambda}) := \sum_{\lambda} P(\mathbf{\Lambda}=\lambda)P_{\text{guess}}(\mathbf{Z}|\mathbf{\Lambda}=\lambda) < 1\,,
\end{equation} where $P_{\text{guess}}(\mathbf{Z}|\mathbf{\Lambda}=\lambda) := \max_\mathbf{z}p(\mathbf{Z}=\mathbf{z}|\mathbf{\Lambda}=\lambda)$. The amount of randomness is quantified by the min-entropy
\ba
H_{\min}(\mathbf{Z}|\mathbf{\Lambda}) := - \log P_{\text{guess}}(\mathbf{Z}|\mathbf{\Lambda})\,.
\ea Clearly, $H_{\min}(\mathbf{Z}|\mathbf{\Lambda})>0$ implies the presence of some randomness. To someone who does not have access to $\mathbf{\Lambda}$, the source will appear to have min-entropy $H_{\min}(\mathbf{Z}) = - \log P_{\text{guess}}(\mathbf{Z})$ which can only be higher by the data processing inequality. Though obvious, it may be worth stressing that $\sum_{\lambda} P(\mathbf{\Lambda}=\mathbf{\lambda}) P_{\text{guess}}(\mathbf{Z}|\mathbf{\Lambda}=\mathbf{\lambda})$ is \textit{not} the same as $P_{\text{guess}}(\mathbf{Z})$, since $P_{\text{guess}}$ is not a given probability distribution but a notation for a procedure that picks up the maximum of a probability distribution. As an extreme example, if $\mathbf{Z}$ looks uniform but the knowledge of $\mathbf{\Lambda}$ determines $\mathbf{z}$ uniquely, one has $P_{\text{guess}}(\mathbf{Z})=1/d^N$ and $\sum_{\lambda} P(\mathbf{\Lambda}=\mathbf{\lambda}) P_{\text{guess}}(\mathbf{Z}|\mathbf{\Lambda}=\mathbf{\lambda})=1$.

The loosest characterization of the source, i.e. the one that requires fewer assumptions, simply puts a bound on the min-entropy:
\begin{definition} \emph{Min-entropy source.} A random variable $\mathbf{Z}$ is a $k$-min-entropy source of randomness with respect to another random variable $\mathbf{\Lambda}$ if $H_{\min}(\mathbf{Z}|\mathbf{\Lambda}) \geq k$.
\end{definition} As soon as $k>0$, the knowledge of $\mathbf{\Lambda}$ does not determine $\mathbf{z}$ uniquely. One can add some structure to a min-entropy source. For instance, a $k$-min-entropy source is called \textit{uniform} if $H_{\min}(\mathbf{Z}|\mathbf{\Lambda}=\lambda) := - \log P_{\text{guess}}(\mathbf{Z}|\mathbf{\Lambda}=\lambda) \geq k$ for all values of $\lambda$. A \textit{block min-entropy source} is one for which not only the min-entropy of the whole string, but the min-entropy of blocks is also lower bounded. These notions will not be used in this paper.

As soon as $k\leq\log (d^N-1)$, the definition of $k$-min-entropy source is compatible with $P_{guess}(\mathbf{Z}=\mathbf{z}|\mathbf{\Lambda}=\lambda)=0$ for one string $\mathbf{z}$. As hinted in paragraph \ref{ss:critical}, the possibility that some settings are not chosen is critical for sources of Bell tests. Because of this, one may want to add to the properties of the source the assumption that \textit{all} the $d^N$ strings have non-zero probability. This is equivalent to the following type of source:

\begin{definition} \emph{Santha-Vazirani sources.} A random variable $\mathbf{Z}$ is a $(p_{min},p_{max})$ Santha-Vazirani source with respect to $\mathbf{\Lambda}$ (where $0 \leq p_{min} \leq 1/d$ and $1/d \leq p_{max} \leq 1$) if  
\begin{equation}
p_{\min} \leq p(z_i|\lambda,z_1,...,z_{i-1}) \leq p_{max} \ \ \forall \ i \ .
\end{equation}
\end{definition} 
If $Z_i$ is a bit, $p_{\min}=1-p_{\max}$ is usually written $\delta$~\cite{SV84}. Some of the most important results in measurement dependence in Bell tests have been obtained for Santha-Vazirani sources~\cite{CR12,GMTDAA12,MP13}. These results show that there is a real advantage in considering Bell-based randomness, because it overcomes no-go theorems for classical information.

Finally, let us focus on distributions that are \textit{independent and identically distributed (i.i.d.)} such that
\begin{equation}
p(\mathbf{Z}=\mathbf{z}|\mathbf{\Lambda}=\lambda)=\prod_{j=1}^N p(Z_j=z_j|\lambda).
\end{equation} This can also be viewed as a block min-entropy source where each block consists of only one symbol, $Z_i$. In this case, the Santha-Vazirani definition implies:
\begin{equation}
p_{\min} \leq p(z|\lambda) \leq p_{max} \ .
\end{equation}
We will use a different notation such that $p_{\max} = P_M$ and $p_{\min} = P_m$ to make clear that we are in the i.i.d.\ scenario. Then the definition of uniform min-entropy sources is equivalent to the figure of merit of measurement dependence used in~\cite{KHSPMKSE12}, namely
\begin{equation}
P_M:=\max_{z,\lambda}p(z|\lambda) \ ,  \qquad \text{[i.i.d.]}
\label{def:Dax}
\end{equation}
since $H_{\min}(Z|\Lambda=\lambda)\geq k \text{ for all } \lambda$ is equivalent to $P_{M}\leq2^{-k}$. In the following, we will use these two figure of merits interchangeably for i.i.d. models.

Instead of bounding the largest probability, the smallest probability also gives information on measurement dependence, as first proposed in~\cite{PHHH09}:\begin{equation}
P_m:=\min_{z,\lambda} p(z|\lambda)  \ .   \qquad \text{[i.i.d.]}
\label{def:lower}
\end{equation}
If only $P_{M}$ is explicitly bounded, then a bound on $P_{m}$ can be inferred, however, it might be trivial, since it can be negative: $P_{m} \geq 1 - (d-1) P_{M}$.  Bounding only the min-entropy of the input source to the Bell test, or equivalently bounding only $P_M$, which is the guessing probability, allows much different worst-case behavior in Bell tests than when the Santha-Vazirani definition is adopted, as we shall now explore.

\section{Lower bound for min-entropy sources}
\label{sec:robustness}

We will be dealing with a $K$-partite Bell scenario where the $i^{\textrm{th}}$ party has $m_i>1$ measurement settings ($m_A,m_B$ for bipartite) and each setting has an arbitrary number of outcomes. The joint configuration of settings $\mathbf{z}=z_1...z_K$ with $z_i\in\{1,..m_i\}$ ($\mathbf{z}=xy$ for bipartite) is a $K$-tuple in the set of all settings $\mathcal{S}$ of size $\prod_{i=1}^{K}m_i$. In this Section, moreover, we consider a Bell test in which the observed statistics of the settings follow a uniform distribution, that is
\begin{equation}
H_{\min}(\mathbf{Z_1,...,Z_K}) = N\log|{\cal S}|\,,
\end{equation} or equivalently
\ba
p_{obs}(z_1...z_K) := \sum_{\lambda} p(z_1...z_K|\lambda)p(\lambda) = \big(\prod_{i=1}^K m_i\big)^{-1}.
\label{eq:pobs}
\ea
This is not an assumption like those on the nature of the source: $p_{obs}$ is observed in a realization; but it is a frequent working assumption for theoretical works, which was made in all previous works on measurement dependence. In Section \ref{sec:nonunif}, we shall see that a non-uniform $p_{obs}$ has interesting consequences in studies of measurement dependence.

We are presently able to discuss our main result: a lower bound on the min-entropy of the source, below which no conclusion can be drawn from any Bell test, unless further structure is assumed.

\subsection{Reaching the no-signaling limit}

The main insight is provided by the following Lemma, which we present in the bipartite scenario (the generalization to multipartite scenarios holds with identical proofs and more cumbersome notation, so we give it in Appendix \ref{appx}):
\begin{lemma}\label{lemma1}
Let $P(ab|xy)$ be an arbitrary no-signaling distribution with $x\in\{1,...,m_A\}$ and $y\in\{1,...,m_B\}$. For any pair of settings $(\bar{x},\bar{y})$, there exists a local distribution $P_L(ab|xy)$ such that
\ba
P_L(ab|xy)&=&P(ab|xy)    \\
\textrm{for } (x,y)&\in&{\cal S}_{\bar{x},\bar{y}}\equiv \left\{(\bar{x},y'),(x',\bar{y}) : x'\in\{1,...,m_A\}\,,  \right. \nonumber \\
& & \qquad \qquad \qquad \qquad \ \ \ \ \ \left. y'\in\{1,...,m_B\} \right\}.  \nonumber 
\label{eq:setdef}
\ea Moreover, this result is tight: if another pair of settings is added to the subset of pairs, there exists a no-signaling point for which those probabilities are nonlocal.
\end{lemma}

\begin{proof}
The proof can be done by constructing explicitly one such local distribution. Let us fix $(\bar{x},\bar{y})=(1,1)$ without loss of generality. From the no-signaling distribution $P$, we construct
\ba
& & \hspace{-2em} \mathbf{P}(a_1,a_2,...,a_{m_A};b_1,b_2,...,b_{m_B})  \nonumber \\
& & \ \  =P(a_1)P(b_1|a_1)\,\prod_{j=2}^{m_A} P(a_j|b_1)\,\prod_{k=2}^{m_B} P(b_k|a_1)
\ea with obvious notations. This is a valid joint probability distribution over the outcomes of all the measurements. Now, on the one hand, the marginals $\mathbf{P}(a_j;b_k)\equiv P_L(a,b|j,k)$ define a local distribution, as first proved by Fine \cite{Fine82}. On the other hand, it is easy to show that $\mathbf{P}(a_1;b_k)=P(a,b|1,k)$: one should sum first over all possible values of $a_2,...,a_{m_A}$ to find $\mathbf{P}(a_1;b_1,b_2,...,b_{m_B})=P(a_1)\,\prod_{k=1}^{m_B} P(b_k|a_1)$, after which the sum over the $b$'s is obvious. Similarly one proves that $\mathbf{P}(a_j;b_1)=P(a,b|j,1)$. So indeed we have a local distribution that mimicks the initial no-signaling one on the desired subset of pairs of settings.

As for the tightness, suppose that we add a single pair of settings, say $(2,2)$, to ${\cal S}_{1,1}$: there exist no-signaling points for which CHSH is violated by the settings $(1,1)$, $(1,2)$, $(2,1)$ and $(2,2)$; so those statistics can't be mimicked by a local distribution.
\end{proof}

Now we can state the main theorem:
\begin{theorem}
Consider a min-entropy source with an observed min-entropy $H_{\min}(\mathbf{XY})=N\log(m_Am_B)$ for an $N$-run bipartite Bell test with $m_A$ inputs on Alice, $m_B$ inputs on Bob and arbitrary alphabets for the outcomes. If
\ba
H_{\min}(\mathbf{XY|\Lambda})&\leq& N\,\log (m_A+m_B-1)
\label{mainbound}
\ea no conclusion can be drawn from the Bell test, since the no-signaling limit of the inequality can be reached with local distributions. The generalization of this result to $K$-partite Bell tests reads
\ba
H_{\min}(\mathbf{Z_1...Z_K|\Lambda})&\leq& N\,\log \Big(\sum_{k=1}^K m_k - K +1\Big)\,.
\label{mainboundmulti}
\ea 
for $H_{\min}(\mathbf{Z_1...Z_K}) = - \log (\prod_K m_K)$.  Notice in particular that, without further assumptions, any source of randomness with $H_{\min}(\mathbf{XY|\Lambda})\leq N\,\log 3$ is useless as a source for any Bell tests.
\end{theorem}

\begin{proof} We will construct an explicit i.i.d.\ source which allows the faking of a Bell violation up to the no-signaling bound with appropriate local resources. From Lemma \ref{lemma1} we know that there exist subsets ${\cal S}_{\bar{x},\bar{y}}$ of $m_A+m_B-1$ pairs of settings, for which no difference can be seen if a local distribution is substituted for a possibly nonlocal no-signaling point: in particular, this could be the no-signaling point that reaches the no-signaling limit for the inequality under study. If $H_{\min}(\mathbf{XY|\Lambda})$ is sufficiently low, the source will allow only the pairs of settings that belong to one of the ${\cal S}_{\bar{x},\bar{y}}$ and distribute the corresponding local strategy $\lambda_{\bar{x},\bar{y}}$. The source
\begin{equation}
p(xy|\lambda_{\bar{x},\bar{y}})=\begin{cases} \frac{1}{m_A+m_B-1}, & \mbox{if } x,y\in\mathcal{S}_{\bar{x},\bar{y}} \\ 0, & \mbox{otherwise }\end{cases}
\label{eq:strat}
\end{equation}
has $P_M= \frac{1}{m_A+m_B-1}$ in each run, whence we have proved the bound \eqref{mainbound} as long as we can find  $p(\lambda_{\bar{x},\bar{y}})$ such that $\sum_{\bar{x},\bar{y}} p(xy|\lambda_{\bar{x},\bar{y}})p(\lambda_{\bar{x},\bar{y}})=p_{obs}(xy)$ for all $x,y$. In the case where $p_{obs}$ is uniform, this can always be found by simply choosing uniformly the pair $(\bar{x},\bar{y})$, i.e. $p(\lambda_{\bar{x},\bar{y}})=\frac{1}{m_Am_B}$. This concludes the proof for the bipartite case. The proof of the multipartite case is identical using the material of Appendix \ref{appx}. The final remark of Theorem~\ref{lemma1} stems from the fact that each Bell test much involve at least two parties and each must have at least two settings.

\end{proof}

Because of the tightness of Lemma \ref{lemma1}, the bounds \eqref{mainbound} and \eqref{mainboundmulti} are the best \textit{inequality-independent bounds} that one can obtain with i.i.d. sources. Moreover, since there exist inequalities for which the quantum and the no-signaling limits coincide, the bound to reach the quantum limit cannot be better. If the inequality is given, however, much less measurement dependence may be sufficient to reach the no-signaling limit, and even less to reach the quantum limit if it is lower. We elaborate further on this point in the following paragraph.

\subsection{Inequality-dependent bounds}

\begin{figure*}
\includegraphics[scale=0.9]{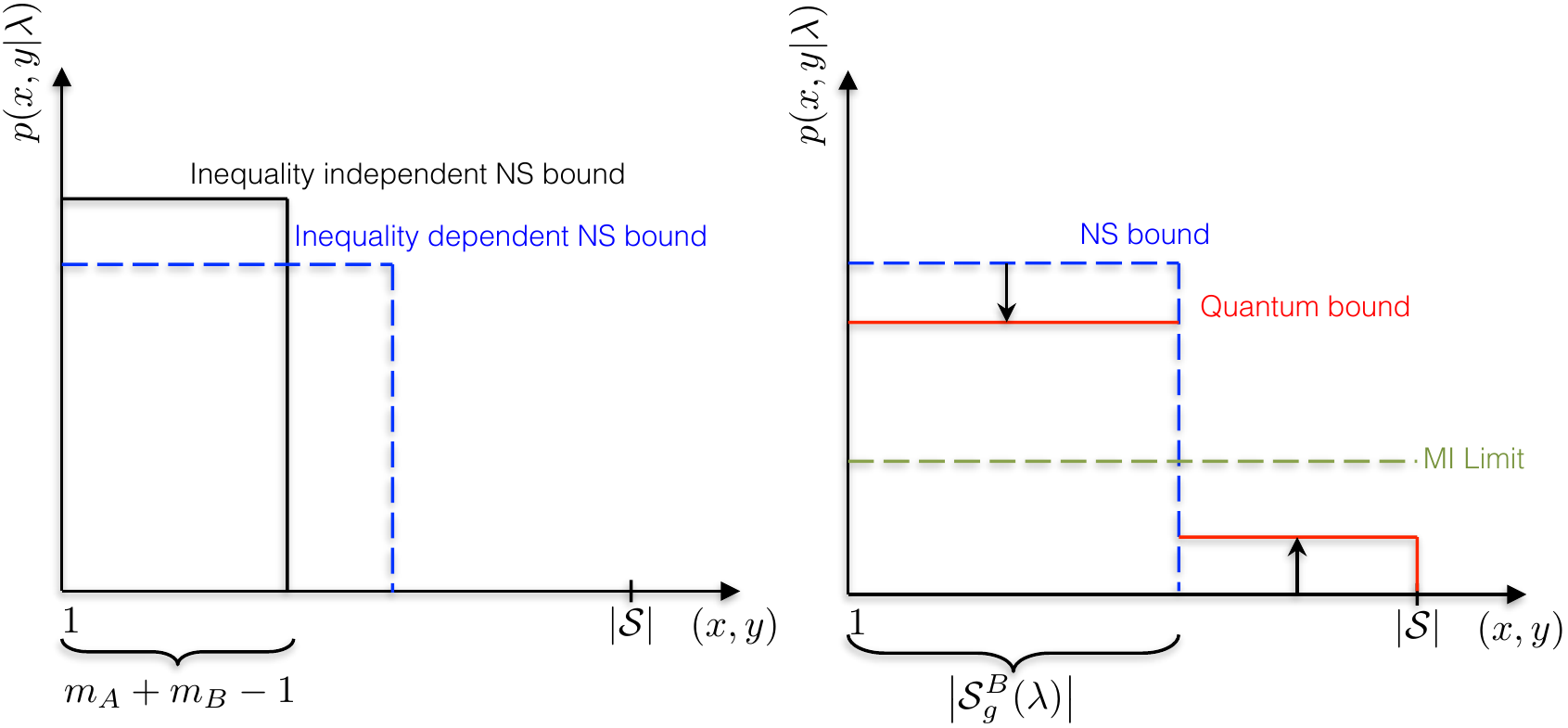}
\caption{(color online)  sources that reach the critical min-entropy bound for uniform observed distribution of settings. For the inequality independent bound \eqref{mainboundmulti}, the source is uniform on $m_A +m_B -1$ settings and is zero elsewhere. For a given inequality, the no-signaling limit may be reached with a source that is uniform on a larger number of settings $|\mathcal{S}^{B}_{g}|$, and still zero on the others; in order to reach only the quantum limit, one can allow the settings in $\mathcal{S}^B_h$ to be used sometimes. Of course, for each $\lambda$, the settings that are chosen may vary.}\label{figsources}
\end{figure*}

Let $B$ define a Bell inequality, whose local, quantum and no-signaling limits are given by $B_\text{L}\leq B_\text{Q}\leq B_\text{NS}$, and $\mathcal{S}^B$ be the set of settings that are used by the Bell inequality~\footnote{The chained Bell inequality is an example of an inequality whose value depends only some of the possible inputs.  Only terms such that $x=y$ or $x=y+1$ and the term for $x=1,\ y=m$ appear.}.  Again, for each $\lambda$ there is a local strategy for assigning outputs such that in order to achieve the no-signaling limit, some settings will be incompatible with this strategy and must be hidden by measurement dependence.  Let this set of inputs be $\mathcal{S}^{B}_{h}(\lambda)$.  Then, an arbitrary no-signaling point is required to be compatible with a local point only on the subset $\mathcal{S}^{B}_{g}(\lambda) := \mathcal{S}^B\setminus\mathcal{S}^B_h(\lambda)$.  Suppose an inspection of the inequality $B$ shows that \textit{at most} $|\mathcal{S}^B_h|$ of these $|\mathcal{S}^B|$ settings must be hidden for any choice of $\lambda$.  Once the probabilities of the settings $\mathcal{S}^B_h(\lambda)$ are set to zero, the min-entropy is maximized by the uniform distribution over the remaining $|\mathcal{S}^{B}_{g}|$ settings (FIG.~\ref{figsources}). However, one must be very careful to show the existence of $p(\lambda)$ which satisfies~(\ref{eq:pobs}). Whenever such a distribution exists, if $P_M\geq 1/|\mathcal{S}^{B}_{g}|$, the non-local game can be won with probability one with local strategies.  As implied by the results of the previous section, if the observed input distribution $p_{obs}(\mathbf{z})$ is uniform then a strategy in the form of equation~(\ref{eq:strat}) or its generalization in Appendix~\ref{appx} with a uniform probability over $\lambda$ will always satisfy~(\ref{eq:pobs}).  However, it is possible to do better in some cases where $|\mathcal{S}^B| < |\mathcal{S}|$.  In such cases $|\mathcal{S}^B_h|$ (the most settings that must be hidden for any $\lambda$) can be small.  Here, for a uniform $p_{obs}$ equation~(\ref{eq:pobs}) will also be satisfied provided possibly more settings than required are hidden for each $\lambda$ such that $|\mathcal{S}^{B}_h(\lambda)| = |\mathcal{S}^{B}_h|$ for all $\lambda$ and if $\mathcal{L}(\mathbf{z})$ is the set of $\lambda$s for which $\mathbf{z} \in \mathcal{S}^{B}_h(\lambda)$, $|\mathcal{L}(\mathbf{z})| = |\mathcal{S}^B_h|$ must be constant for all $\mathbf{z}$.  This is a symmetry condition that can be met by many Bell inequalities.  As before, the existence of this example proves that a min-entropy source with 
\ba
k \leq  N\,\log \big(|\mathcal{S}^{B}_{g}|\big)\label{boundnsineq}
\ea can reach the no-signaling limit of $B$ with local strategies for uniform input distributions.  In the following section we will show how to obtain bounds for arbitrary $p_{obs}$ and that approach will also give tight bounds and optimal strategies when the inequality is one in which the size of the ``hidden sets'' varies with $\lambda$.

Further, if $B_\text{Q}< B_\text{NS}$, in order to simulate physics one may be content with \textit{reaching the quantum limit}. A possible i.i.d.\ source (not proved to be optimal) is the following (see Fig.~\ref{figsources}). With probability $1-q$, the settings are chosen uniformly among all ${\cal{M}}$ possible $K$-tuples: this is measurement independence, so $B\leq B_L$ on these cases, and the physical process $\lambda$ can be chosen as one of those that saturate $B=B_\text{L}$. In the other instances, the settings are chosen uniformly in $\mathcal{S}^{B}_{g}$ and the physical process $\lambda$ is chosen in each case in order to achieve $B=B_\text{NS}$. In other words, this source is a convex combination of the measurement independent uniform source and the source described in the previous paragraph. Note that this new source will automatically satisfy the constraint ($\ref{eq:pobs})$. For such a source, therefore, $P_M$ is the probability of each setting in $\mathcal{S}^{B}_{g}$, which reads $P_M=\frac{1}{|\cal{S}|}\left[1+q\left({|\cal{S}|}/|\mathcal{S}^{B}_{g}|-1\right)\right]$. With this measurement-dependent strategy, one can reach $B=qB_\text{NS} + (1-q)B_L$, so $B\geq B_\text{Q}$ for $q\geq (B_\text{Q}-B_\text{L})/(B_\text{NS}-B_\text{L})$. In summary, the quantum limit can be achieved with an i.i.d.~source with
\ba
P^{B,Q}_M&\geq& \frac{1}{|\cal{S}|}\left[1+\frac{B_\text{Q}-B_\text{L}}{B_\text{NS}-B_\text{L}}\left({|\cal{S}|}/|\mathcal{S}^{B}_{g}|-1\right)\right]\,,
\label{eq:PQ}
\ea that is, a min-entropy source with $k \leq  -N\,\log P^{B,Q}_M$ can reach the quantum limit of $B$ with local strategies, for a uniform input distribution.

Let us illustrate the methodology with the analysis of some inequalities:
\begin{itemize}
\item \textit{CHSH}: here, it is always necessary and sufficient to hide one pair of settings. Therefore $|\mathcal{S}^{B}_{g}|=3$ and the inequality-dependent bound \eqref{boundnsineq} is the same as the inequality-independent one \eqref{mainbound} to reach the no-signaling limit, as already proved in \cite{KHSPMKSE12}. Recall that this does not prove the bounds to be tight, because they are based on explicit i.i.d. sources: non i.i.d. sources may lead to tighter bounds, though we do not know any example. As for reaching the quantum limit, we have $P^{B,\text{Q}}_M=\frac{1}{4}[1+(\sqrt{2}-1)/3]\approx 0.2845$. 

\item \textit{Chained inequality}: here again, as we have seen in paragraph \ref{ss:critical}, it is always necessary and sufficient to hide only one pair of settings out of ${\cal{M}}=m^2$, so $|\mathcal{S}^{B}_{g}|=m^2-1$ and $|S^{B}_h(\lambda)|=1$ for all $\lambda$. As a consequence, in terms of min-entropy, the inequality-dependent bound \eqref{boundnsineq} is $N\,\log (m^2-1)$, which is approximately twice the value $N\,\log (2m-1)$ obtained from \eqref{mainbound}. For large $m$, the quantum and no-signaling limits basically coincide.

\item \textit{CGLMP inequalities}:  like the CHSH inequality, the CGLMP inequalities are two party inequalities where each party has two inputs.  However, this family of inequalities has $d$ possible outputs for each party.  In the quantum case, the CGLMP inequalities can provide more robustness against measurement dependence than the CHSH inequality, in the sense that the min-entropy of the inputs given the source must be lower if the quantum bound is to be achieved.  The reason is that it has been shown that as $d\rightarrow \infty$, the quantum limit increases and approaches the no-signaling limit~\cite{ZG08,ZRGW10}.  As can be seen, inspecting equation~(\ref{eq:PQ}), the value of $(B_\text{Q}-B_\text{L})/(B_\text{NS}-B_\text{L})$ will increase with $d$, and the value of $P_{M}$ necessary to reach the quantum limit with local resources increases, until it reaches the no-signaling value $P_M^{B,\text{NS}}$ in the limit.

\item \textit{Mermin inequalities}:  Mermin inequalities~\cite{M90,BK93} are multipartite inequalites such that for odd numbers of parties, the quantum and no-signaling bounds coincide.  For this reason, the 5-party Mermin inequality was used in~\cite{GMTDAA12} to amplify randomness.  When the number of parties is an odd number at least 3 only a subset of all possible inputs appear in the corresponding Mermin inequality and the inequality-independent bound is not tight.  In general, for odd $K$ parties $|\mathcal{S}^{B}_{g}|=2^{K-2}+2^{(K-3)/2}$ and $|\mathcal{S}^{B}|=2^{K-1}$~\cite{CGR08}. Specifically for the 5-party case, $|\mathcal{S}^{B}_{g}|=10$ and $|\mathcal{S}^{B}|=16$.

\end{itemize}

\section{The positive effect of biasing the choices of the settings}
\label{sec:nonunif}

Theorem 1 shows that assuming a full min-entropy source on the measurement settings, for any meaningful conclusion to be drawn from a Bell test, it must be that $H_{\min}(\mathbf{XY|\Lambda})>N\,\log (m_A+m_B-1)$. However, recalling that the role of the observed data is actually a \textit{constraint} imposed on the underlying model (similar to equation (\ref{eq:pobs})), we can hope to use it to our advantage. This motivates the question: for a given value of $H_{\min}(\mathbf{XY|\Lambda})=N k>N\,\log (m_A+m_B-1)$ that is being assumed, what is the optimal distribution on the inputs such that the maximum possible Bell value obtainable with this degree of measurement dependence and only local resources is as low as possible.  Because the situation for non i.i.d. models is intractable, we are restricting ourselves to the i.i.d. model for the remaining of this chapter. Here instead of the min-entropy, the guessing probability $P_M$ is used exclusively as the figure of merit of measurement dependence. First, we consider the CHSH inequality as an explicit example.

\subsection{The CHSH Inequality}

Intuitively, we expect that the optimal solution is to set for each input round $ H_{\min}(XY)=H_{\min}(XY|\Lambda) = k$ and $p_{obs}(xy) = 2^{-k}$ for three pairs $(x,y)$ and $p_{obs}(x'y') = 1-3\times2^{-k}$ for the final pair because in this case $\Lambda$ cannot contain any further information on $XY$ than is available simply from observing the distribution $p_{obs}(xy)$.  We will highlight an example of this type of distribution later in this section. This is not a uniform distribution, so we can already see that non-uniform input distributions can be beneficial.  In this section, we will consider fixed input distributions $p_{obs}(xy)$ and find the maximum value that the CHSH inequality can take given a bound on $P_M$.  Note that the method in this section extends to any multipartite Bell inequality.

We want to find the violation $B_{\text{CHSH}}^{\max}$, under local resources and measurement dependence, as a function of $P_M$ and $p_{obs}(xy)$. To this end, observe that the local distributions form a convex polytope and so is the set of sources with a fixed value of $P_M$ (the source polytope). Using the decomposition into extremal points of a convex polytope,  we have
\ba
p(ab|xy\lambda)&=&\sum_{i}\alpha_{i}(\lambda)e_{i}(ab|xy)\,,\\
p(xy|\lambda)&=&\sum_{j}\beta_{j}(\lambda)f_{j}(xy)\,,
\ea
where $e_{i}(ab|xy)$ are the extremal points of the local polytope and $f_{j}(xy)$ are the extremal points of the source polytope. Now after multiplying by both sides by $p_{obs}(xy)$, the i.i.d. model with measurement dependence becomes 
\begin{eqnarray}
p(abxy) &=& \int_\Lambda d\lambda \sum_{ij}\alpha_{i}(\lambda)\beta_{j}(\lambda)e_{i}(ab|xy)f_{j}(xy)p(\lambda) \nonumber \\
&=& \sum_{ij}\gamma_{ij}g_{ij}(abxy)\,, 
\end{eqnarray}
where
\begin{eqnarray}
\gamma_{ij} = \int_\Lambda d\lambda \alpha_{i}(\lambda)\beta_{j}(\lambda)p(\lambda)\,,  \\
g_{ij}(abxy) = e_{i}(ab|xy)f_{j}(xy)\,.
\end{eqnarray}
In this notation, the problem becomes a linear program, i.e. finding
\begin{equation}
B_{\text{CHSH}}^\text{max}(P_M,p_{obs}(xy))=\max_{p(abxy)} \ \ \sum_{abxy}(-1)^{a+b+xy}\frac{p(abxy)}{p_{obs}(xy)}
\end{equation}
subjected to the constraints
\begin{equation}
p(abxy)=\sum_{ij}\gamma_{ij}g_{ij}(abxy), \ \sum_{ab}p(abxy)=p_{obs}(xy)
\end{equation}
for known values of $p_{obs}(xy)$ and $P_M$. The result is presented in FIG.~\ref{fig:CHSHnonunfinputs}.

\begin{figure}[h]
\includegraphics[trim=1cm 1cm 1cm 1cm, scale=0.47]{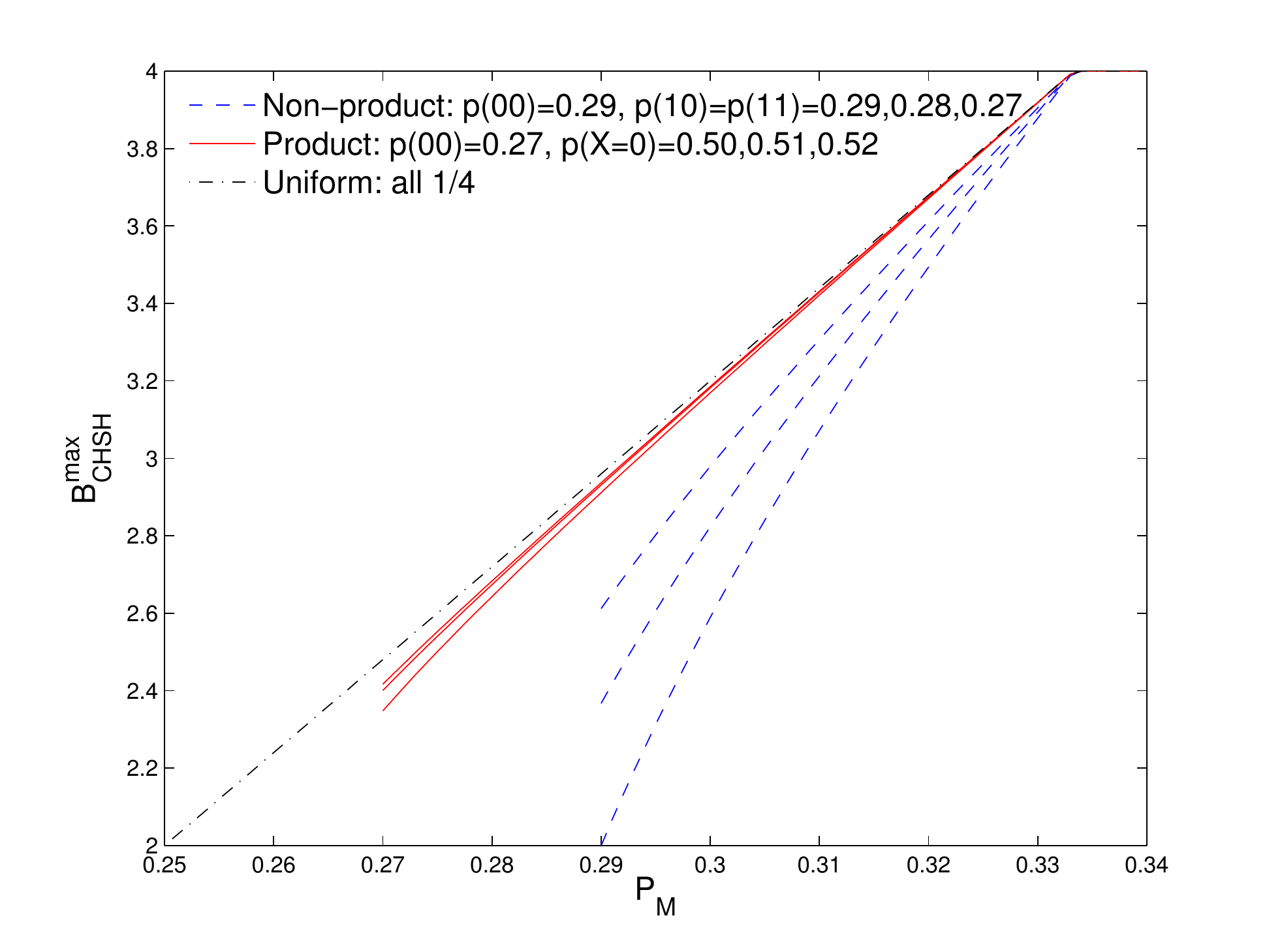}
\caption{(color online) Plot of maximum CHSH value against measurement dependence $P_M$ for different $p_{obs}$. Notice that $P_M$ start from $\max\{p_{00},p_{01},p_{10},p_{11}\}$ because of the data processing inequality: no underlying model of smaller $P_M$ can reproduce the observed input statistics.  In a Bell test with an assumed dependence bound $P_M$, if the value of the inequality is above the line that corresponds to the observed input distribution $p_{obs}$ then there is intrinsic randomness in the outcomes contributed by $\lambda$.  Therefore, for some observed violations, biased settings statistics can allow the certification of intrinsic randomness while uniform statistics cannot.}
\label{fig:CHSHnonunfinputs}
\end{figure}

Using the numerical results, it is easy to see that the optimal strategy for maximizing the Bell value $B_{\text{CHSH}}$ whether or not the observed distribution is uniform is to choose 
\ba
p(xy|\lambda_{\bar{x},\bar{y}}) = \left\{ \begin{array}{cc} P_M & \text{if } x,y\in \mathcal{S}_{\bar{x},\bar{y}}  \\  1-3P_{M} & \text{otherwise}  
\end{array} \right.  \ ,
\ea
for each $\lambda_{\bar{x},\bar{y}}$, where $\mathcal{S}_{\bar{x},\bar{y}}$ is defined in equation~(\ref{eq:setdef}).  Choosing this strategy, it is straightforward to find an analytic expression for $B_{\text{CHSH}}^{\max}$:
\begin{equation*}
B_{\text{CHSH}}^{\max}(P_M) = 4 - \frac{1}{2}\left[ (1-3P_M)q + \frac{1-3P_M}{4P_M-1} (q -16) \right] \,
\end{equation*}
where for convenience we define $q: = \sum_{x,y} \frac{1}{p_{obs}(x,y)}$.  This expression is only valid for $\max_{x,y} p_{obs}(xy)\leq P_M < \frac{1}{3}$ ($H^{obs}_{\min} (xy) \geq H_{\min}(xy|\lambda) > \log 3$).  Notice that when the distribution is uniform $q = 16$ and the second term vanishes, leaving a linear expression in $P_M$.

It is interesting to observe that for the purpose of violating Bell inequalities (that is, demonstrating non-locality by exceeding $B^{\max}$) under measurement dependence, suppose the inputs have privacy quantified by $P_M$, then it is advantageous for us to \emph{purposely} select an input distribution that is not uniform. This can be seen easily from for example the red curves for $P_M = 0.27$: selecting uniform input distribution allows a violation up to about 2.5 while selecting non-uniform input distribution only allows a lower maximum violation!  Note that for non-uniform distributions on the inputs the upper bound on the Bell value is only as low as 2 (the local bound assuming measurement independence) for non-product distributions on the inputs.  (See the blue dashed curves.)  All non-uniform product input distributions can have Bell values larger than 2, if measurement independence is relaxed.  Notice also, that the lowest blue curve, the one that takes the value 2 at $P_M = 0.29$ is the one corresponding to the distribution $[p(00), p(01), p(10), p(11)] = [0.29, 0.13, 0.29, 0.29]$.  This is precisely the form of the distribution on the inputs we has anticipated at the start of this section.

\subsection{Generalizations}

We have seen that for the case of the CHSH inequality, the strategy outlined in section~\ref{sec:robustness} in equation~(\ref{eq:strat}) is the optimal strategy even in the case that the distribution $p_{obs}(xy)$ is not uniform.  In general however this is not the case.  It is possible to find some inequalities that together with some distributions $p_{obs}(\mathbf{z})$ do not admit a strategy of the form
\begin{equation}
p(xy|\lambda_{\bar{\mathbf{z}}})=\begin{cases} P_M, & \mbox{if } \mathbf{z}\in\mathcal{S}_{\bar{\mathbf{z}}} \\ Q(\lambda_{\bar{\mathbf{z}}}), & \mbox{otherwise }\end{cases}
\label{eq:strat2}
\end{equation}
where $Q(\lambda)$ is determined by the normalization condition to be $Q(\lambda_{\bar{\mathbf{z}}}) = \frac{1-|S_g^B(\lambda)|}{|S_h^B(\lambda)|}$.  

Let us limit our focus to inequalities with symmetries such that $|S_g^{B}(\lambda)| = |S_g^{B}|$ and $|S_h^{B}(\lambda)| = |S_h^{B}|$ for all $\lambda$.  In that case, equation~(\ref{eq:pobs}) can be written as a matrix equation, with $p_{obs}(\mathbf{z})$ and $p(\lambda)$ written as vectors and $p(\mathbf{z}|\lambda_{\bar{x},\bar{y}})$ is a matrix whose entries are defined by equation~(\ref{eq:strat2}).  If the $p(\mathbf{z}|\lambda)$ matrix is is full-rank, then there is a unique solution for $p(\lambda)$ that is a valid probability distribution.  This will always be the case if $|S_g|$ and $|S_h|$ have no common factors.  

Examples of cases where the sizes of the sets $S_g$ and $S_h$ have no common factors are any bipartite Bell inequality with terms for all input pairs present and where both parties have the same number of inputs.  For these cases, the min-entropy bound of section~\ref{sec:robustness} also applies for any non-uniform observed distribution on the inputs.

If, for a given inequality, $|S_g|$ and $|S_h|$ have at least one common prime factor, there may be some choices of distribution $p_{obs}$ for which the strategy~(\ref{eq:strat2}) will not be able to reproduce $p_{obs}$ with any valid distribution $p(\lambda)$.  In that case, the optimal strategy may have to be found numerically.  For the i.i.d.\ case, one do this by solving a linear program that is a generalization of the one presented in the previous section.

\section{Bounds on the achievable distributions}
\label{sec:lambda}

While the set of distributions obtainable from a measurement independence local model is the local polytope, that obtainable from a measurement dependence model is in principle a larger set. The example in section \ref{ssviol} makes it clear that this  set can even include signaling distributions. Now we wish to study more carefully this new set of distributions. As mentioned in \ref{sec:measind}, figures of merit of measurement dependence can be defined as restrictions on $p(\lambda|xy)$ instead of $p(xy|\lambda)$. It turns out that in characterizing the set of achievable distributions, it is easier to work with figures of merits based on $p(\lambda|xy)$. In general, specifying $p(xy|\lambda)$ does not specify $p(\lambda|xy)$ unless $p(xy)$ and $p(\lambda)$ are uniform, in which case the two are proportional.  Among the figures of merit one can define is the one by Hall~\cite{H11}.  We shall adopt a slightly modified measure:
\begin{definition}
The quantity $M'$ bounds the distance between the probability distribution over the random variable $\Lambda$ given access to $x,y$ and the distribution over $\Lambda$ without information on $x,y$:
\begin{equation}
M' := \max_{x,y} 2D\big(p(\Lambda|X=x,Y=y), p(\Lambda)\big) \,,
\label{eq:mydef}
\end{equation}
with $p(\Lambda) = \sum_{x,y} p(\Lambda|X=x,Y=y) p_{obs}(X=x,Y=y)$ and $D(\cdot, \cdot)$ is the total variational distance. 
\end{definition}
Essentially this definition bounds the distance of any one element in the distribution away from a strategy independent of Alice and Bob's inputs $x$ and $y$, which is similar in flavour to a Santha-Vazirani bound.

Now we wish to compare the point obtainable from a measurement dependence model with local resources,
\begin{equation}
p^\lambda_{ab|xy} = \sum_{\lambda} p(a|x \lambda) p(b|y \lambda) p(\lambda|xy)  \ ,
\label{eq:hvmodel}
\end{equation} 
to its corresponding measurement independent point,
\begin{equation}
p_{ab|xy} = \sum_{\lambda} p(a|x \lambda) p(b|y \lambda) p(\lambda)  \ .
\label{eq:hvmodel}
\end{equation} 
Their distinguishability is characterized by their total variational distance
\begin{eqnarray}
& & \hspace{-3em} D(p^{\lambda}_{AB|xy}, p_{AB|xy})  \nonumber \\
&=& \frac{1}{2} \sum_{a,b} |p^{\lambda}_{ab|xy} - p_{ab|xy}|  \nonumber \\
&=& \frac{1}{2} \sum_{a,b} \left|\sum_{\lambda} p(a|x \lambda) p(b|y \lambda) \left( p(\lambda|xy) - p(\lambda)\right)\right|   \nonumber \\
&\leq& \frac{1}{2}  \sum_{\lambda} \left(\sum_{a,b} p(a|x \lambda) p(b|y \lambda)\right) |p(\lambda|xy) - p(\lambda)|   \nonumber \\
&=& D\big(p(\lambda|x,y), p(\lambda)\big)    \nonumber \\
&\leq& M'/2 \, ,
\end{eqnarray}
where we have applied the triangle inequality. In other words,
\begin{equation}
| p_{ab|xy}^\lambda - p_{ab|xy} |  \leq M' \ \ \forall \  a,b,x,y \,,
\label{eq:probmd}
\end{equation}
which gives a bound on the distributions that can be created with a measurement dependence model of dependence bounded by $M'$.  
 
There is another model that relaxes assumptions about correlations between two parties: quantum ``cross-talk'' between non-fully-isolated devices~\cite{SPM12}.  In this case, the assumption that the measurement operators are in a tensor product, $M_A \otimes M_B$, is relaxed and thus it is possible to use the measurement itself to introduce quantum correlations between the two parties.  The model was proposed to describe excess correlations that could result from a pair of trapped ions measured while situated next to each other in a refrigerator after being entangled.   Comparing our equation~(\ref{eq:probmd}) with equation~(2) in~\cite{SPM12}, we find that the distributions allowed by cross-talk are formally the same as those allowed by measurement dependence as defined in equation~(\ref{eq:mydef}), even though the physical interpretation is very different.

\section{Conclusions}

Bell tests are an essential tool in device-independent approaches.  They rely on a set of reasonable assumptions, but some of the assumptions are untestable.  In particular, the correlations between source and settings are \emph{strictly unobservable} and therefore the amount of reduction of measurement independence is ultimately an assumption, either on the power of an adversary, on a physical model for the experiment.  This study has demonstrated that when relaxing this assumption, the definition used, be it min-entropy or a Santha-Vazirani condition, is critical with respect to what kind of guarantees can be obtained from a Bell test.  There are results~\cite{GMTDAA12,GHHHPR13} showing that with a Santha-Vazirani source assumption arbitrarily weak randomness can be amplified using a protocol that checks for the violation of a Bell inequality.  This cannot be accomplished using a min-entropy condition, as we have demonstrated in section~\ref{sec:robustness}: for sufficiently low min-entropy any inequality can be violated up to its no-signaling bound, using only the classical measurement dependent correlations and in a way that a third party could predict all of the outcomes of the measurements.  Even for the protocol in~\cite{CR12} that amplifies bounded randomness (in the Santha-Vazirani definition) using violations of the chained Bell inequality, in order to get perfectly free bits out, the number of outputs for this inequality must go to infinity.  As we point out in section~\ref{ss:critical}, in this limit the chained Bell inequality is not robust to any relaxation of input randomness if the min-entropy definition is used instead.

The bounds on the min-entropy presented in section~\ref{sec:robustness} give immediate bounds for any inequality on the amount of input randomness required to draw conclusions about whether the violation of a Bell inequality can give any certification of quantum or non-local behavior.  The method demonstrated for the CHSH inequality in section~\ref{sec:nonunif} demonstrates how to get tight upper bounds for the value a given Bell inequality can take assuming a min-entropy bound for any distribution over the measurement settings.  It also shows that there may be advantages to deliberately choosing non-uniform distributions over measurement settings in device independent protocols, depending on what assumptions are being made.  Relaxing the assumption of measurement-dependence increases the set of probability distributions $\{p(ab|xy)\}$ that can result from a Bell test assuming a local, realistic hidden variable model and section~\ref{sec:lambda} gives an expression bounding this increase.

Being as the assumption of measurement independence cannot be confirmed, it is important to understand the consequences for device independent protocols when it is relaxed.  It is especially interesting that the min-entropy condition, a condition widely adopted in classical security studies~\cite{Vadhan,ILL89,DS02,DOPS04}, is has such a different behavior from the Santha-Vazirani condition for these device-testing purposes.  We hope that the bounds and characterizations provided here will be useful for constructing protocols that are more robust to extraneous correlations.

\acknowledgements{This work is funded by the Singapore Ministry of Education (partly through the Academic Research Fund Tier 3 MOE2012-T3-1-009) and the Singapore National Research Foundation.  We would like to thank Jean-Daniel Bancal, Jeysthur Ang, Roger Colbeck, Yaoyun Shi, Aarthi Sundaram and Miklos Santha for helpful discussions.}

\bibliographystyle{unsrt}	
\bibliography{randomness}

\appendix

\section{Generalization of Lemma \ref{lemma1} the multipartite case}
\label{appx}

The bound \eqref{mainboundmulti} for the min-entropy in a multipartite scenario is based on the generalization of Lemma \ref{lemma1} that we provide here:

\begin{lemma}\label{lemma2}
Let $P(\mathbf{o}|\mathbf{z})$ be an arbitrary $K$-partite no-signaling distribution with $\mathbf{z}=(z_1,...,z_K)$ where $z_i\in\{1,...,m_i\}$ and $\mathbf{o}$ is a $K$-tuple of outcomes. For any $K$-tuple of settings $\bar{\mathbf{z}}=(\bar{z}_1,...,\bar{z}_K)$, there exists a local distribution $P_L(\mathbf{o}|\mathbf{z})$ such that
\ba
P_L(\mathbf{o}|\mathbf{z}) &=& P(\mathbf{o}|\mathbf{z})   \\
\textrm{for }  \mathbf{z} &\in& {\cal S}_{\bar{\mathbf{z}}}\equiv \left\{(\bar{z}_1,z'_2,...,z'_K),..., (z'_1,...,z'_{K-1},\bar{z}_K)  :   \right.   \nonumber \\
& & \qquad \qquad \qquad \qquad \left.  z'_i\in\{1,...,m_i\} \right\}\,. \nonumber
\ea 
Moreover, this result is tight: if another $K$-tuple of settings is added to the subset ${\cal S}_{\bar{\mathbf{z}}}$, there exist a no-signaling point for which those probabilities are nonlocal.
\end{lemma}

\begin{proof}
Again, let us fix $\bar{\mathbf{z}}=(1,...,1)$ without loss of generality.  Let $o^i_{z_i}$ be the $i$th party's outcome given the $z_i$th measurement setting.  From the no-signaling distribution $P$, we construct a valid probability distribution
\begin{widetext}
\ba
\mathbf{P}(o^1_{1}...o^1_{m_1};...;o^K_{1}...o^K_{m_K})&=&P(o^1_1)\left[\prod_{i=2}^{K} P(o^i_1|o^1_1...o^{i-1}_1)\right]\,\left[\prod_{i=1}^{K}\prod_{j=2}^{m_i} P(o^i_j|o^1_1...o^{i-1}_1 o^{i+1}_1...o^{K}_1)\right]
\label{eq:gendistr}
\ea 
\end{widetext}
whose marginals $\mathbf{P}(o^1_{z_1};...;o^K_{z_K})\equiv P_L(o^1...o^K|z_1...z_k)$ define a local distribution by Fine's result \cite{Fine82}. To verify that we have a local distribution that mimics the initial no-signaling one on the desired subset of pairs of settings, consider this example: for the input string $(1,z_2,...z_k)$ with the distribution $\mathbf{P}(o^1_{1};o^2_{z_2};...;o^K_{z_K})=P_{L}(o^1o^2...o^K|1,z_2,...z_k)$ we sum first over all possible values of each outcome variable $o^1_2,...,o^1_{m_1}$ to find
\begin{widetext}
\ba
\mathbf{P}(o^1_1;o^2_{1}...o^2_{m_2};...;o^K_{1}...o^K_{m_K})= P(o^1_1) \, \prod_{i=2}^{K} \left[P(o^i_1|o^1_1...o^{i-1}_1)\,\prod_{j=2}^{m_i} P(o^i_j|o^1_1...o^{i-1}_1 o^{i+1}_1...o^{K}_1)\right]
\ea
\end{widetext}
after which continue to sum over all the $o^2_k,...,o^K_k$ except $o^2_{z_2},...,o^K_{z_K}$ and one is left with a probability distribution $\mathbf{P}(o^1_{1},o^{2}_{z_2}...o^K_{z_K})$ on only $K$ variables, one for each party. The other verifications are similar.  Another way to think of it is to notice that each conditional probability factor on $K$ variables (one variable conditioned on $K-1$ other variables) effectively sets a joint probability distribution on those same $K$ variables.  In the distribution~(\ref{eq:gendistr}) there are $\sum_{i=1}^{K} m_i - K+1$ such factors and so this is exactly how many local points $P_L(\mathbf{o}|\mathbf{z})$ that can be matched for a given hidden variable value (see equation~(\ref{mainboundmulti}) in the main text).  The argument for tightness still works if we consider only two parties among $K$.  For any two parties we can choose a pair of inputs for each to return to a CHSH-type scenario, then the argument follows in the same way as in the proof of Lemma~\ref{lemma1}.

\end{proof}

\end{document}